\documentclass[letterpaper, 10 pt, conference]{ieeeconf} 
\IEEEoverridecommandlockouts                              
\usepackage{setspace}
\usepackage{amssymb}
\usepackage{amsmath}
\usepackage{color}
\newtheorem{definition}{Definition}
\newtheorem{assumption}{Assumption}
\newtheorem{theorem}{Theorem}
\newtheorem{remark}{Remark}
\newtheorem{lemma}{Lemma}
\usepackage{cite}
\usepackage{algorithm}
\usepackage{algorithmicx}
\usepackage[noend]{algpseudocode}
\usepackage{bm}
\usepackage{graphicx}
\usepackage{footmisc}
\usepackage{epstopdf}
\usepackage{geometry}
\allowdisplaybreaks[4]
\makeatletter
\newenvironment{breakablealgorithm}
  {
   \begin{center}
     \refstepcounter{algorithm}
     \hrule height.8pt depth0pt \kern2pt
     \renewcommand{\caption}[2][\relax]{
       {\raggedright\textbf{\ALG@name~\thealgorithm} ##2\par}%
       \ifx\relax##1\relax 
         \addcontentsline{loa}{algorithm}{\protect\numberline{\thealgorithm}##2}%
       \else 
         \addcontentsline{loa}{algorithm}{\protect\numberline{\thealgorithm}##1}%
       \fi
       \kern2pt\hrule\kern2pt
     }
  }{
     \kern2pt\hrule\relax
   \end{center}
  }
\makeatother

\title{\LARGE \bf
Privacy-Preserving Push-sum Average Consensus via State Decomposition
}
\author{Xiaomeng Chen$^{1}$, Lingying Huang$^{1}$, \thanks{$^{1}$X. Chen, L. Huang and L. Shi are with the Department of Electronic and Computer Engineering, Hong Kong University of Science and Technology, Clear Water Bay, Kowloon, Hong Kong ( email:xchendu@connect.ust.hk, lhuangaq@connect.ust.hk, eesling@ust.hk).}  Kemi Ding$^{2}$, \thanks{$^{2}$K. Ding is with the School of Electrical and Electronic Engineering, Nanyang Technological University, Singapore 639798, Singapore (email: kemi.ding@ntu.edu.sg).} Subhrakanti Dey$^{3}$,\thanks{$^{3}$S. Dey is with the Hamilton Institute, NUI Maynooth, Maynooth, Co. Kildare, Ireland, and also affiliated with Uppsala University. (email: subhrakanti.dey@angstrom.uu.se).} and Ling Shi$^{1}$
}

\begin{document}
\maketitle
\thispagestyle{plain}
\pagestyle{plain}

\begin{abstract}
Average consensus is extensively used in distributed networks for computation and control, where all the agents constantly communicate with each other and update their states in order to reach an agreement. Under a general average consensus algorithm, information exchanged through wireless or wired communication networks could lead to the disclosure of sensitive and private information. In this paper, we propose a privacy-preserving push-sum approach for directed networks that can protect the privacy of all agents while achieving average consensus simultaneously. Each node decomposes its initial state arbitrarily into two substates, and their average equals to the initial state, guaranteeing that the agent's state will converge to the accurate average consensus. Only one substate is exchanged by the node with its neighbours over time, and the other one is reserved. That is to say, only the exchanged substate would be visible to an adversary, preventing the initial state information from leakage. Different from the existing state-decomposition approach which only applies to undirected graphs, our proposed approach is applicable to strongly connected digraphs. In addition, in direct contrast to offset-adding based privacy-preserving push-sum algorithm, which is vulnerable to an external eavesdropper, our proposed approach can ensure privacy against both an honest-but-curious node and an external eavesdropper. A numerical simulation is provided to illustrate the effectiveness of the proposed approach. 
\end{abstract}

\section{INTRODUCTION}
With increasing applications in smart grids, smart buildings and intelligent transportation systems, etc, the cooperative distributed algorithm has been a heated research topic during the last decade. When all the components of a network reach a common agreement, we say that the distributed system reaches a consensus. One of the most commonly adopted consensus algorithm is the average consensus algorithm, where each agent aims to reach the average of their initial values. The convergence of average consensus is firstly proved by DeGroot \cite{degroot1974reaching} and further studied by other researchers (e.g., \cite{chatterjee1977towards}, \cite{tsitsiklis1986distributed}). Typical applications of average consensus include distributed sensor fusion \cite{xiao2005scheme}, load balancing in parallel computing \cite{boillat1990load} and coordinated control \cite{ren2005consensus}. 

Conventional average consensus algorithm requiring each node to exchange their state information with the neighbouring nodes in order to reach the average consensus, is not desirable if the participating nodes have sensitive and private information. In addition, by hacking into communication links, an external eavesdropper has access to state information, which is exchanged through wireless or wired communication networks. As the number of privacy leakage events is increasing, there is an urgent need to preserve privacy of each agent in distributed systems. 

  Several approaches have been proposed in recent years to protect privacy. The main idea of most existing privacy-preserving approaches is to mask signals by adding noises. Nozari et al.\cite{NOZARI2017221} proposed a differentially private consensus algorithm by adding some uncorrelated noises. However, it cannot converge to the exact average due to the tradeoff between the privacy and accuracy.  To improve this tradeoff, Mo et al.\cite{mo2016privacy} devised a new mechanism where exchanged information is masked by correlated noises, and the convergence to the correct average is also guaranteed. 
Another strand of research emerged recently is observability-based privacy-preserving approaches. Alaeddini et al. \cite{alaeddini2017adaptive} guaranteed the privacy protection by minimizing the information about a certain node from the observability perspective.\\ 
\indent None of the aforementioned approach, however, is suitable for directed graphs with weak topological restrictions. To preserve privacy of nodes interacting on an unbalanced graph, Charalambous et al. \cite{charalambous2019privacy} proposed an offset-adding privacy-preserving approach, and Gao et al. \cite{gao2018privacy} protected privacy by adding randomness on edge weights, both of which are only effective against honest-but-curious nodes. To improve resilience to external eavesdroppers, Hadjicostis et al. \cite{hadjicostis2020privacy} employed homomorphic encryption to maintain privacy, relying on a trusted node. As a consequence, this approach requires  a large amount of computation and communication, which may be inapplicable for systems with limited resources. Aiming at protecting privacy against both honest-but-curious nodes and eavesdroppers, Wang \cite{wang2019privacy} proposed a privacy-preserving mechanism in which the state of a node is decomposed into two-substates. However, it is inapplicable to directed graphs and only considers the eavesdropper which is unaware of the entire information about the network. To summarize, the main challenges arising in the design and analysis of privacy-preserving average consensus algorithm are as follows:
\begin{enumerate}
	\item \textbf{Communication over directed graphs:} In practice, the information flows among sensors may not be bidirectional due to the different communication ranges, e.g., the coordinated vechicle control problem \cite{ghabcheloo2005coordinated} and the economic dispatch problem \cite{yang2013consensus}. Owing to the unbalanced interactions over agents, the approaches in \cite{NOZARI2017221},\cite{mo2016privacy} and \cite{wang2019privacy} fail to achieve average consensus in such scenarios. 
	\item \textbf{Definition of privacy-preserving:} The privacy notion of differential privacy is not suitable for the privacy-preserving mechanism which is not based on noise injection. Therefore, a new privacy notion is needed to measure the privacy degree.
	\item \textbf{Existence of eavesdroppers:} Different from honest-but-curious nodes, an external eavesdropper is a stronger adversary which has access to more information. It is difficult to protect privacy against an external eavesdropper with low computation load.
\end{enumerate}


 In this paper, based on the concepts introduced in \cite{wang2019privacy}, we propose a state-decomposition based privacy-preserving algorithm for directed networks that can maintain the privacy of all agents while guaranteeing the accuracy of the average consensus at the same time. In addition, we specify the estimation strategy of the eavesdropper which knows more information than \cite{wang2019privacy} and provide analytical results on its estimation performance. The main contributions of this paper are summarized as follows:
\begin{enumerate}
\item We propose a novel privacy-preserving push-sum algorithm for strongly connected digraphs, which addresses explicitly the constraints imposed by the topology of the communication network (\textbf{Algorithm \ref{alg2}}). Furthermore, using coefficients of ergodicity \cite{seneta2006non}, we prove the convergence of our proposed approach to the exact average of the initial value (\textbf{Theorem 1}).
\item Different from the privacy notion in \cite{charalambous2019privacy} and \cite{hadjicostis2020privacy} which only considers the exact initial value,  we define the privacy preservation against honest-but-curious nodes, where the privacy of each node is preserved if honest-but-curious nodes have infinite uncertainty (see \textit{Definition 4} for more details) on the initial value based on the accessible information. Moreover, we prove that our proposed approach can preserve privacy of each node against honest-but-curious nodes for certain topological conditions (\textbf{Theorem 2} and \textbf{Theorem 3}).
\item  We analyze the privacy-preserving performance of the proposed algorithm, in the presence of an eavesdropper (see \textit{Definition 2}) and prove that the estimation error of the eavesdropper cannot be bounded in probability (\textbf{Theorem 4}). In contrast, the estimation error converges to zero under the existing privacy-preserving approaches in \cite{charalambous2019privacy} and \cite{gao2018privacy}.
\end{enumerate}

\textit{Notations:} In this paper,  $\mathbb{N}$ and $Z_+$ represent the sets whose components are  natural numbers and positive integers. $\mathcal{N}(\mu, \sigma^2)$ denotes the Gaussian distribution with mean $\mu$ and covariance $\sigma^2$. $U(a,b)$ denotes the uniform distribution over the interval $(a,b).$ For an arbitrary vector $\bm{x},$ we denote its $i$th element by $\bm{x}_i$. For an arbitrary matrix $\bm{M},$ we denote its  element in the $i$th row and $j$th column by $[\bm{M}]_{ij}$. $\lfloor x \rfloor$ denotes the greatest integer less than or equal to $x$. Finally, ``w.p. $p$'' stands for ``with probability $p$'' in this paper.
\section{PRELIMINARIES}
\subsection{Network Model}
We consider a directed graph (digraph) $G\triangleq (\mathcal{V}, \mathcal{E})$ with $N$ nodes, where $\mathcal{V}=\{1,2,\ldots,N\}$ denotes the node set and $\mathcal{E}\subset \mathcal{V} \times \mathcal{V}$ denotes the edge set, respectively. A communication link from node $i$ to node $j$ is denoted by $(j,i)\in  \mathcal{E}$, indicating that node $i$ can send messages to node $j$. The self-loop is not included in the directed graph $G$. The nodes who can directly send messages to node $i$ are represented as in-neighbours of node $i$ and the set of these nodes is denoted as $N_i^{in}=\{j\in \mathcal{V}\mid (i,j)\in \mathcal{E}\}$. Similarly, the nodes who can directly receive messages from node $i$ are represented as out-neighbours of node $i$ and the set of these nodes is denoted as $N_i^{out}=\{j\in \mathcal{V}\mid (j,i)\in \mathcal{E}\}$. In this paper, we consider that the initial value of each node is not equal to the average of them\footnote{{ The probability of the event that the initial value of a node equals to the average value of all nodes is zero if each initial value is a continuous random variable.}}, but each node runs the proposed privacy-preserving push-sum algorithm (Algorithm 3) to infer their average.
  { \begin{assumption}\label{asp1}
The digraph $G$ is assumed to be strongly connected with $N$ nodes, where $N >2$. In other words, there exists at least one directed path from any node $i$ to any node $j$ in the digraph with $i\neq j$. 

 \end{assumption}}

\subsection{General Push-sum Algorithm}
The push-sum algorithm, introduced originally in \cite{kempe2003gossip}, aims to achieve average consensus for each node communicating on a directed graph with relatively weak topological restrictions. Consider a network of $N$ nodes, where each node has a private initial state, {\color{black}termed as  $x_i(0)$ for node $i\in \mathcal{V}$.} Without loss of generality, we assume that the initial state is a scalar. In the push-sum algorithm, each node generates two values, $x_{i,1}(k)$ and $x_{i,2}(k)$, both of which are updated in the same way. The algorithm is described as follows:\begin{breakablealgorithm}
\caption{General Push-sum Algorithm}
\label{alg1}
\begin{algorithmic}
 \State \textbf {Step 1}.   Node $i\in \mathcal{V}$ initializes $x_{i,1}(0)=x_i(0)$ and $x_{i,2}(0)=1$. The coupling weight between node $j$ and node $i$ is denoted as $p_{ji}$ and the self-weight of node $i$ is denoted as $p_{ii}$. 
 \State \textbf {Step 2}.  At iteration $k$:\\
 \begin{enumerate}
 \item Node $i$ randomly chooses a set of weights, $\{p_{ji}(k) \sim U({\color{black}0},1)\mid j\in N_i^{out}\cup\{i\}\}$,  { and then normalizes them such that $\sum_{j=1}^N p_{ji}(k)=1$}. Also, $ p_{ji}(k) $ is set to be 0 if $(j, i)\notin N_i^{out}$. 
  \item Node $i$ computes $p_{ji}(k)x_{i,1}(k)$ and $p_{ji}(k)x_{i,2}(k)$, and sends them to its out-neighbors $j\in N^{out}_i$.  
\item After receiving the information from its in-neighbors $j\in N^{in}_i$, node $i$ updates $x_{i,l}$ as follows:
$$
  x_{i,l}(k+1)=\sum\limits_{j\in N^{in}_i\cup\{i\}}p_{ij}(k)x_{j,l}(k),\qquad  l=1,2.
$$

\item Node $i$ computes the estimated average  $$
 { \hat x^{ave}_i}(k+1)=  x_{i,1}(k+1)/ x_{i,2}(k+1).$$

\end{enumerate}
\end{algorithmic}
\end{breakablealgorithm}

It is known that the general push-sum algorithm can reach the exact average if the network is strongly connected and the matrix $\textbf{P}(k)$ is column-stochastic, where $\textbf{P}(k)=[p_{ij}(k)]$ with $p_{ij}(k)$ defined in Algorithm \ref{alg1} \cite{rezaeinia2019push}.

\subsection{Privacy Leakage}
In this paper, we consider two types of adversaries, which are defined as follows. 
\begin{definition}\label{df3}
An  honest-but-curious adversary is a node which follows the system's  protocol and attempts to infer the private information of other nodes under the knowledge of its received data.
In addition, it can collude with other honest-but-curious nodes, i.e., multiple honest-but-curious nodes can share their received data with each other.
\end{definition}

{ We further consider a class of eavesdroppers.}
\begin{definition}
\label{df2}
{\color{black}An  eavesdropper is an external attacker who  has no prior knowledge of the privacy-preserving method the system adopted but knows the network topology and weights between every pair of nodes. Furthermore, it can intercept all transmitted data, { and
runs Algorithm \ref{alg3} to estimate the private information of each node.}}
\end{definition}

Generally speaking, an eavesdropper is more disruptive than an honest-but-curious node since it has access to more information. Let $\mathcal{I}_e(k)$ denote the information set available to an external eavesdropper at iteration $k$, given by: 
\begin{equation*}
\label{infor}
\mathcal{I}_e(k)=\{p_{ij}(k),p_{ij}(k)x_{j,l}^+(k)\mid \forall i,j\in \mathcal{V} ,j\neq i, l=1,2\},
\end{equation*}
{where $x_{j,l}^+(k)$ denotes the obfuscated version of $x_{j,l}(k)$}. With the accumulated information, an eavesdropper can estimate the initial value of each node by the following algorithm, which shares a similar concept { to eavesdropper estimating algorithms  for undirected graphs in \cite{wang2019privacy} and \cite{ruan2019secure}}.  
 \begin{breakablealgorithm}
\caption{Eavesdropper Estimating Algorithm}
\label{alg3}
\begin{algorithmic}
 \State \textbf {Step 1}. The external eavesdropper initializes $s_1(0)=x_{i,1}^+(0)$ and $s_2(0)=x_{i,2}^+(0)$. 
 \State \textbf {Step 2}.  At iteration $k$:\\
 \begin{enumerate}
  \item The eavesdropper computes $p_{ii}(k)=1-\sum\limits_{j\neq i}p_{ji}(k)$, and updates $s_1(k)$  and $s_2(k)$ as follows:	
 \vspace*{-0.5mm}
$$
 s_l(k+1)=s_l(k)+x_{i,l}^+(k+1)-\sum\limits_ {j\in N_{i}^{in}\cup \{i\}} p_{ij}(k)x_{j,l}^+(k),
 \vspace*{-2mm}
$$

where $l=1,2$.
\item The initial value of node $i$ is estimated by
$$
\hat x_i^0(k)=s_1(k)/s_2(k).
$$
\end{enumerate}
\end{algorithmic}
\end{breakablealgorithm}

{ \begin{remark}
		Algorithm 2 constructs an effective observer for an eavesdropper to estimate the initial value of node $i$, which mimics the system under the general push-sum Algorithm \ref{alg1}. By adopting Algorithm 2, the eavesdropper can successfully estimate the initial value of each node which runs the general push-sum algorithm or existing privacy-preserving approaches in  \cite{charalambous2019privacy} and \cite{gao2018privacy} {\color{black}as $\lim\limits_{k\rightarrow\infty} \hat x_i^0(k)=x_i(0)$, which is shown in Section \ref{sim}.} For different average consensus algorithms, the estimation algorithms of the eavesdroppers would be different accordingly.  To the best of our knowledge, Algorithm 2 in our paper is the first proposed algorithm that is effective for directed graphs. 
		\end{remark}}

\section{PRIVACY-PRESERVING ALGORITHM}
In this section, we propose a state-decomposition based  push-sum algorithm on directed graphs to preserve privacy (as it has already shown in \cite{gao2018privacy} that the general push-sum algorithm leads to the leakage of privacy) and converge to the accurate average simultaneously. 
\subsection{Privacy-preserving Push-sum Algorithm}
 The main idea of our approach is to let each node decompose its state $x_{i,l}(k)$ into two substates $x^\alpha_{i,l}(k)$ and $x^\beta_{i,l}(k)$, $l=1,2$. The substate  $x^\alpha_{i,l}(k)$ is exchanged with other nodes while $x^\beta_{i,l}(k)$ is never shared with other nodes (i.e., is reserved by node $i$ itself). Hence, the substate $x^\beta_{i,l}(k)$ is imperceptible to the neighbouring nodes of node $i$. 
 \begin{breakablealgorithm}

\caption{Privacy-perserving Push-sum Algorithm}
\label{alg2}
{\color{black}\begin{algorithmic}
 \State \textbf {Initialization:}\\
 \begin{enumerate}
 	\item Node $i\in \mathcal{V}$ { randomly generates an initial substate value  $x_{i,1}^\alpha(0)$ from $U(-M,M)$, where $M>0$ is a pre-defined value, and initializes $x_{i,1}^\beta(0)=2x_{i}(0)-x_{i,1}^\alpha(0)$, $x_{i,2}^\alpha(0)=0$ and $x_{i,2}^\beta(0)=2$, where $x_i(0)$ denotes the private initial state of node $i$. } 
 \end{enumerate}
  \State \textbf {Weight generation:}\\
  \begin{enumerate}
 \item For $k=0$, node $i\in \mathcal{V}$ randomly chooses a set of weights,  $\{p_{ji}(0), \alpha_i(0) \mid j\in N_i^{out}\cup \{i\}\}$ from  $\mathcal{N}(0, M)$, and then normalizes them such that $\sum_{j=1}^N p_{ji}(0)+\alpha_i(0)=1$. Also, $ p_{ji}(k) $ is set to be 0 if $j\notin N_i^{out}$. 
 \item For $k\geq1$, node $i$ randomly chooses a set of weights, {  $\{p_{ji}(k), \alpha_i(k)\sim U ({\color{black}0},1)\mid j\in N_i^{out}\cup \{i\}\}$,  and normalizes them such that $\sum_{j=1}^N p_{ji}(k)+\alpha_i(k)=1$.} Also, $ p_{ji}(k) $ is set to be 0 if $j\notin N_i^{out}$. 
  \end{enumerate}
\State \textbf {State update:} For all $k\geq 0$\\
 \begin{enumerate}
   \item Node $i\in \mathcal{V}$ computes $p_{ji}(k)x^\alpha_{i,1}(k)$ and $p_{ji}(k)x^\alpha_{i,2}(k)$, and sends them to its out-neighbours $j\in N^{out}_i$.  
\item After receiving the information from its in-neighbors $j\in N^{in}_i$, node $i$ updates its two substates $x_{i,l}^\alpha(k+1)$ and $ x_{i,l}^\beta(k+1)$ as follows:
\begin{equation}
\label{e1}
  \left\{  
  \begin{aligned}
  &x_{i,l}^\alpha(k+1)=\sum\limits_{j\in N^{in}_i\cup\{i\}}p_{ij}(k) x^\alpha_{j,l}(k)+ x^\beta_{i,l}(k),\\
  & x^\beta_{i,l}(k+1)=\alpha_i (k)x^\alpha _{i,l}(k),   
  \end{aligned}
  \right. 
\end{equation}
with $i\in \mathcal{V}$, $l=1,2$.

\item Node $i$ computes the estimated average $$ { \hat x^{ave}_i}(k+1)=x^\alpha_{i,1}(k+1) /x^\alpha_{i,2}(k+1).$$

\end{enumerate}
\end{algorithmic}}
\end{breakablealgorithm}

 \subsection{Convergence Analysis}
In this subsection, we prove the  convergence of Algorithm \ref{alg2}.

Under Algorithm \ref{alg2}, as illustrated in Fig. \ref{virtual},  we can view the substate $x_{i,l}^\beta(k)$, $\l\in\{1,2\}$ as the state of a \textit{virtual node} $i^{\beta}$ which can only communicate with the actual node $i$. In other words, all substates compose a  digraph $G'$ with $2N$ nodes, and the sum of the initial values of all nodes in $G'$ is twice of that in the original graph $G$.
\begin{figure}[h]
    \centering
           
         \vspace*{-2mm}\includegraphics[width=0.21\textwidth]{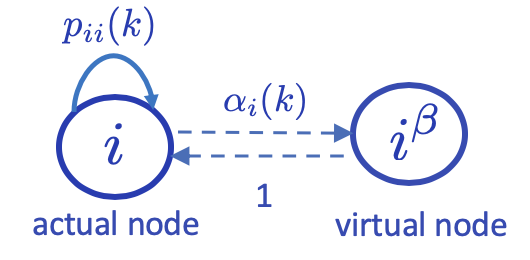}
    \caption{Virtual node $i^\beta$ and corresponding weights. }
            \vspace*{-5mm}
    \label{virtual}
    \end{figure}
    
	
	Using matrix-vector notation, the iteration rule in (\ref{e1}) can be rewritten as follows:
\begin{equation}
\label{e2}
  \left\{  
  \begin{aligned}
   &\bm{x_1}(k+1)=\bm{\hat P}(k)\bm{x_1}(k), \\
   & \bm{x_2}(k+1)=\bm{\hat P}(k)\bm{x_2}(k), 
  \end{aligned}
  \right. 
\end{equation}\vspace*{-2mm}
where $$\vspace*{-2mm}\bm{x_1}(k)=[x_{1,1}^\alpha(k),\ldots,x_{N,1}^\alpha(k), x_{1,1}^\beta(k),\ldots,x_{N,1}^\beta(k)]^\top,$$ $$\bm{x_2}(k)=[x_{1,2}^\alpha(k),\ldots,x_{N,2}^\alpha(k), x_{1,2}^\beta(k),\ldots,x_{N,2}^\beta(k)]^\top,$$ 
\begin{equation*}  
  \bm{\hat P}(k)=\left[
	\begin{array}{cc}
		\bm{P}(k)&\bm{I}_{n\times n}\\ 
		\bm{\Lambda}(k)_{n\times n} & \bm{O}
	\end{array}
	\right],
	\end{equation*}
	with $\bm{P}(k)=[p_{ij}(k)]$ , $\bm{\Lambda}(k)=diag(\alpha_1(k),\ldots,\alpha_N(k))$ (see definitions of $p_{ij}(k)$ and $\alpha_{i}(k)$ in Algorithm \ref{alg2}).
Then, we can obtain
\begin{equation}
  \left\{  
  \begin{aligned}
   &\bm{x_1}(k+1)=\bm{\hat P}(k) \cdots\bm{\hat P}({ 1})\bm{x_1}({ 1}), \\
   & \bm{x_2}(k+1)=\bm{\hat P}(k) \cdots\bm{\hat P}({ 1})\bm{x_2}({ 1}). 
  \end{aligned}
  \right. 
\end{equation}

 {\color{black}Let $\bm{T}_k$ denote the product 
\begin{equation}
\bm{T}_k=\bm{\hat P}(k)\bm{\hat P}(k-1) \cdots\bm{\hat P}(1).	
\end{equation}

We can easily check that matrix $\bm{T}_k$ is column-stochastic since it is the product of column-stochastic matrices.
Next, we use the coefficient of ergodicity \cite{seneta2006non} to establish the convergence of Algorithm \ref{alg2}.

\begin{definition}\label{co}
	For a column stochastic matrix $P_c$, the coefficient of ergodicity $\delta(P_c)$ is defined as
		$$\delta(P_c)=\max\limits_{j}\max\limits_{i_1,i_2}|P_c(j, i_1)-P_c(j,i_2)|.
		\vspace*{2mm}
		 $$
\end{definition}

From Definition \ref{co}, it can be seen that  $\delta(P_c)$  characterizes how different two columns of $P_c$ are; in particular, $\delta(P_c)=0$ if and only if the columns of $P_c$ are identical.

\begin{lemma}
	Suppose \textit{Assumption \ref{asp1}} holds. The coefficient of ergodicity $\delta(\bm{T}_k)$ converges almost surely to zero.
\end{lemma}
\begin{proof}
	Define $$W_t=\prod_{k=N(t-1)+1}^{Nt}{\bm{\hat P}}(k), \quad \forall t\in\mathbb{Z_+}, $$
	where $N$ is the number of actual nodes in the network. Note that  there are paths from each virtual node $i^\beta$ to its corresponding actual node  $i$ and the actual nodes in $\mathcal{V}$ are strongly connected. Hence, the digraph corresponding to the matrix $\bm{\hat P}(k)$ (diagonal excluded) has paths from all nodes to the actual nodes for any $k \in\mathbb{Z_+}$. Furthermore, since $\bm{\hat P}(k)$ has positive diagonal entries at  the location of the actual nodes for any $k \in\mathbb{Z_+}$, any node $i$ or $i^\beta$ has at least one path of length $N$ to any actual node in $\mathcal{V}$.
	
	Thus,  there is at least one row with all entries strictly positive in $W_t, \forall t \in\mathbb{Z_+}$. Let $\alpha_t >0$ be the minimum value of the entry in such a row, then we have $\alpha_t\geq\epsilon^N$ for all $t\in\mathbb{Z_+}$, where $\epsilon= \min_{i, j\in\mathcal{V}: p_{ij}(k)>0} p_{ij}(k).$ 
	
	Next, following the standard results on coefficients of ergodicity (see, e.g. \cite{seneta2006non}, and \cite{hajnal1958weak}), we can obtain that if {\color{black}$y^\top=x^\top W_t$}, then
	$$\max\limits_{i} y_i- \min \limits_{i}y_i \leq (1-\alpha_t) (\max\limits_{i} x_i- \min\limits_{i} x_i).$$
	
	Meanwhile, for a generic column stochastic matrix $P_c$, we have that if {\color{black}$y^\top=x^\top P_c,$} then $\max_{i} y_i- \min _{i}y_i \leq (\max_{i} x_i- \min_{i} x_i).$  The forward product $\bm{T}_k=\bm{\hat P}(k)\bm{\hat P}(k-1) \cdots\bm{\hat P}(1)$ can be assembled in blocks of length $N$, i.e., {\color{black}$$\bm{T}_k= \prod\limits_{t=k+1-(k\mod N)}^{k} \bm{\hat P}(t)\prod\limits_{t=1}^{\lfloor k/N \rfloor} W_t ,$$}
	where $\prod_{t=k+1}^{k} \bm{\hat P}(t)=1. $
	 Iterating, we have that if {\color{black}$y^\top=x^\top\bm{T}_{k}$}, then 
	$$\begin{aligned}
	\max\limits_{i} y_i- \min \limits_{i}y_i 		 &\leq \prod_{t=1}^{\lfloor k/N \rfloor}(1-\alpha_t) (\max\limits_{i} x_i- \min\limits_{i} x_i)\\
	&\leq (1-\epsilon^N)^{\lfloor k/N \rfloor} (\max\limits_{i} x_i- \min\limits_{i} x_i).\\
		\end{aligned}$$
		
		By varying $x$ among the vectors of the canonical  basis, we have 
		$$\max\limits_{i} [\bm{T}_{k}]_{ji}-\min\limits_{i}[\bm{T}_{k}]_{ji}\leq (1-\epsilon^N)^{\lfloor k/N \rfloor}, \forall j \in{1,2,\ldots, 2N},$$
		i.e., $\delta(\bm{T}_{k})\leq (1-\epsilon^N)^{\lfloor k/N \rfloor}.$ As $k \rightarrow \infty$, $(1-\epsilon^N)^{\lfloor k/N \rfloor}\rightarrow  0$ happens with probability $1$ and then we have $\delta(\bm{T}_k) \xrightarrow{a.s.}0$, which finishes the proof.
		\end{proof}}

The following theorem shows that the privacy-preserving push-sum algorithm  converges to the average value of all initial values. 
\begin{theorem}
	{ For a digraph satisfying \textit{Assumption \ref{asp1}}}, under Algorithm \ref{alg2}, the estimated average ${ \hat x^{ave}_i}(k+1)$ will converge to the average of all initial values $\sum_{i=1}^N x_i(0)/N$ with probability one. Also, $ 
	x^\beta_{i,1}(k+1)/ x^\beta_{i,2}(k+1)$ will converge to the average with probability one. 
\end{theorem}
\begin{proof}
	 {Since $\bm{\hat P}(0)$  is column stochastic, we have
	\begin{equation*}
	\begin{aligned}	
		&\bm{1}^\top \bm{x_1}(1)=\bm{1}^\top \bm{\hat P(0)}\bm{x_1}(0)=\bm{1}^\top \bm{x_1}(0)=2\sum_{i=1}^N x_i(0),\\
		&\bm{1}^\top \bm{x_2}(1)=\bm{1}^\top \bm{\hat P(0)}\bm{x_2}(0)=\bm{1}^\top \bm{x_2}(0)=2N. 
			\end{aligned}
	\end{equation*}

	}

 {\color{black}Lemma 1 implies that $\delta(\bm{T}_k) \xrightarrow{a.s.}0$, i.e., $\bm{T}_k$ tends to have identical columns with probability one. Thus, we have $\bm{\hat P}(\infty) \cdots\bm{\hat P}(1)=\bm{v1}^\top$, where $\bm{v}=[v_i]$ is a stochastic vector. }Therefore, the estimated average can be obtained by:
\begin{equation}
\begin{aligned}\label{cov1}
	{ \hat x^{ave}_i}(\infty)&=\frac{x^\alpha_{i,1}(\infty)}{ x^\alpha_{i,2}(\infty)} =\frac{[\bm{\hat P}(\infty) \cdots\bm{\hat P}(1)\bm{x_1}({ 1})]_i}{[\bm{\hat P}(\infty) \cdots\bm{\hat P}(1)\bm{x_2}({ 1})]_i}\\
	&=\frac{[\bm{v1}^\top\bm{x_1}({ 1})]_i}{[\bm{v1}^\top\bm{x_2}({ 1})]_i}=\frac{v_i(\bm{1}^\top\bm{x_1}({ 1}))}{v_i(\bm{1}^\top\bm{x_2}({ 1}))}\\
	&=\frac{2\sum_{i=1}^N x_i(0)}{2N}=\frac{\sum_{i=1}^N x_i(0)}{N} \quad \text{w.p.$1$}.
\end{aligned}
\end{equation}
Furthermore,
\begin{equation}
\begin{aligned}
	&\frac{x^\beta_{i,1}(\infty)}{ x^\beta_{i,2}(\infty)}=\frac{[\bm{\hat P}(\infty) \cdots\bm{\hat P}(1)\bm{x_1}({ 1})]_{i+N}}{[\bm{\hat P}(\infty) \cdots\bm{\hat P}(1)\bm{x_2}({ 1})]_{i+N}}\\
	&=\frac{[\bm{v1}^\top\bm{x_1}({ 1})]_{i+N}}{[\bm{v1}^\top\bm{x_2}({ 1})]_{i+N}}=\frac{v_{i+N}(\bm{1}^\top\bm{x_1}({ 1}))}{v_{i+N}(\bm{1}^\top\bm{x_2}({ 1}))}\\
	&=\frac{2\sum_{i=1}^N x_i(0)}{2N}=\frac{\sum_{i=1}^N x_i(0)}{N} \quad \text{w.p.$1$}. 
\end{aligned}
\end{equation}
\end{proof}
 \begin{remark}
 It is worth noting that we assume that the system $G$ with $N$ nodes is expanded to the system $G'$ with	$2N$ nodes for the convenience of convergence analysis. In fact, the number of system nodes does not change. Thus, there is no need to build an additional communication structure under the proposed algorithm, which is desirable if the communication resource is limited. 
 \end{remark}

\subsection{Privacy-preserving Performance Analysis Against Honest-but-Curious Nodes}
In this subsection, we prove that Algorithm \ref{alg2} protects privacy against honest-but-curious nodes. 

{\color{black}According to \textit{Definition 1}, we consider a set of honest-but-curious nodes $\mathcal{A}$ aiming to infer the initial value of node $i \in \mathcal{L}$ based on the information accessible to it, where $\mathcal{L}=\mathcal{V}\backslash \mathcal{A}$ denotes the set of legitimate nodes.	Under Algorithm \ref{alg2}, the information set accessible to the set of honest-but-curious nodes $\mathcal{A}$ at time $k$ can be defined as 
	\begin{equation}\label{info}
	\mathcal{I}_\mathcal{A}(k)=\{\mathcal{I}_a(k) \mid  a\in\mathcal{A}\},\vspace*{-3mm}
	\end{equation}
	where $$\vspace*{-1mm}
	\begin{aligned}
	\mathcal{I}_a(k)\triangleq  \{&x_{a,l}^\alpha(k), x_{a,l}^\beta(k),p_{ja}(k), p_{ap}(k)x_{p,l}^\alpha(k)\\&\mid p\in N_a^{in}, j\in \mathcal{V},l=1,2\}.
	\end{aligned}
	$$

	Given time instant $\kappa \in \mathbb{N}$, {\color{black}the honest-but-curious nodes $\mathcal{A}$ obtain such set of information sequence $\mathcal{I}_\mathcal{A}(0:\kappa)=\cup_{1\leq k\leq \kappa}\mathcal{I}_\mathcal{A}(k)$. For any feasible set $\mathcal{I}_\mathcal{A}(0:\kappa)$, the term  {\color{black}$\Delta(\mathcal{I}_\mathcal{A}(0:\kappa),i)$} denotes  the set of all initial values $x_i(0)$ at node $i$ that
	there exists a set of $x_{n,1}^{\alpha}(0), x_{n,1}^{\beta}(0), n\in  \mathcal{L}$, and a
	sequence of $p_{jn}(k), j\in \mathcal{V}, k=0,1,\ldots, \kappa$ such that the sequence of
	$x_{a,l}^\alpha(k), x_{a,l}^\beta(k),p_{ja}(k), p_{ap}(k)x_{p,l}^\alpha(k), a\in\mathcal{A}, p\in N_a^{in}, j\in \mathcal{V},l=1,2$  generated by Algorithm \ref{alg2} is equal to that in the adversary information set
	$\mathcal{I}_A(0:\kappa)$.}

	The set {\color{black}$\Delta(\mathcal{I}_\mathcal{A}(0:\kappa),i)$} includes all possible initial  values of node $i$ that can generate $\mathcal{I}_\mathcal{A}(0:\kappa)$ in  (\ref{info}). The diameter of {\color{black}$\Delta(\mathcal{I}_\mathcal{A}(0:\kappa),i)$} is defined as 
	$$\text{Diam}(\mathcal{I}_\mathcal{A}(0:\kappa))=\sup\limits_{x_i(0), x_i(0)'\in\ \Delta(\mathcal{I}_\mathcal{A}(0:\kappa),i)} |x_i(0)-x_i(0)'|.$$

	{   \begin{definition}\label{df5}
			The privacy of node $i \in\mathcal{L}$ is preserved against a set of honest-but-curious nodes $\mathcal{A}$ if, for any $\kappa \in \mathbb{N}, \text{Diam}(\mathcal{I}_\mathcal{A}(0:\kappa))=\infty$ for any feasible $\mathcal{I}_\mathcal{A}(0:\kappa)$.
	\end{definition}}
		
	{\color{black}Definition \ref{df5} shares a similar idea to the uncertainty-based privacy notion in \cite{lu2020privacy}, which is inspired from the  notion of $l$-diversity \cite{machanavajjhala2007diversity}. In  $l$-diversity, the diversity of the discrete-valued sensitive data is measured by the  number of different valuations for the data, and a larger diversity leads to a larger uncertainty on the sensitive data. In our problem, we view the continuous-valued $x_i(0)$ as the sensitive data, whose diversity is measured by the diameter the set {\color{black}$\Delta(\mathcal{I}_\mathcal{A}(0:\kappa),i)$}. A larger diameter represents a larger diversity/uncertainty. }
	%
	
	Note that in our problem setting, we say that the privacy is preserved if the diameter of the set {\color{black}$\Delta(\mathcal{I}_\mathcal{A}(0:\kappa),i)$} is infinite for any feasible $\mathcal{I}_\mathcal{A}(0:\kappa)$ for any $\kappa \in \mathbb{N}$, achieving the largest possible diversity. That is to say, the adversary is not able to find a unique value or even a meaningful range of $x_i(0)$ and it is more stringent than the privacy definition in \cite{mo2016privacy} and \cite{gao2018privacy} which adopt the noise-adding approach to preserve privacy, because it is possible that the estimated value  resides in a range around the true initial value. }{\color{black}Without loss of generality, we consider the scenario on how to protect the initial value of node $i$.}
\begin{theorem} 
	For a digraph satisfying \textit{Assumption \ref{asp1}}, the privacy of any node $i$ can be preserved against    a set of honest-but-curious nodes $\mathcal{A}$ under Algorithm 3 if {\color{black}$N_i^{out} \cup N_i^{in}\nsubseteq \mathcal{A}$}. 
\end{theorem}
\begin{proof}   {\color{black}Since $N_i^{out} \cup N_i^{in}\nsubseteq \mathcal{A}$, there exists at least one node $m$ that belongs to $N_i^{out} \cup N_i^{in}$ but not $\mathcal{A}$.} Fix any $\kappa \in \mathbb{N}$ and any feasible information set $\mathcal{I}_\mathcal{A}(0:\kappa)$. We denote {\color{black} $\{x_{n,1}^\alpha(0)', x_{n,1}^\beta(0)',p_{jn}(k)'\mid n\in \mathcal{L}, j\in \mathcal{V}, k=0,1,\ldots \kappa\}$ }as an arbitrary set of initial substate values and  weights that satisfies $\mathcal{I}_\mathcal{A}(0:\kappa)$. Hence, {\color{black}we have\footnote{{\color{black}To simplify notation, we write $\Delta(\mathcal{I}_\mathcal{A}(0:\kappa),i)$ as $\Delta(\mathcal{I}_\mathcal{A}(0:\kappa))$ in the remaining paper.}} $x_i(0)'\in \Delta(\mathcal{I}_\mathcal{A}(0:\kappa),i)$,} where $x_i(0)'=(x_{i,1}^\alpha(0)'+x_{i,1}^\beta(0)')/2.$ We then denote $x_i(0)''$ as $x_i(0)''=x_i(0)'+e$, where $e$ is an arbitrary real number.

		Next we show that there exists a set of values $\{x_{n,1}^\alpha(0)'', x_{n,1}^\beta(0)'', p_{jn}(k)'', n \in \mathcal{L}, j\in \mathcal{V}, k=0,1,\ldots, \kappa\}$ which makes $x_i(0)''\in \Delta(\mathcal{I}_\mathcal{A}(0:\kappa)).$ (Note that the valuation of the initial value of node $i$ should still guarantee the convergence to the original average  after altering $x_m(0)'$ to $x_m(0)'-e$, i.e., the sum of all nodes' initial values does not change.) The initial substate values $x_{n,1}^\alpha(0)'', x_{n,1}^\beta(0)''$ are denoted as follows, which satisfy $\hat x_{n,1}^\alpha(0)''+\hat x_{n,1}^\beta(0)''=2\hat x_n(0)'', \forall n\in \mathcal{L}$.
	
	\begin{equation}\label{sub}
	\begin{aligned}
	&x_{i,1}^\alpha(0)''= x_{i,1}^\alpha(0)',  x_{i,1}^\beta(0)''= x_{i,1}^\beta(0)'+2e, \\
	&x_{m,1}^\alpha(0)''= x_{m,1}^\alpha(0)', x_{m,1}^\beta(0)''= x_{m,1}^\beta(0)'-2e,\\
	&x_{q,1}^\alpha(0)''= x_{q,1}^\alpha(0)', x_{q,1}^\beta(0)''= x_{q,1}^\beta(0)',\forall q\in  \mathcal{L} \backslash \{i,m\}.\\
	\end{aligned}
	\end{equation}
	Then we divide the derivation into two situations: $m\in N^{in}_i$ and $m\in N^{out}_i$.
	
	Situation \uppercase\expandafter{\romannumeral 1}: If $m\in N^{in}_i$,  a set of coupling weights is considered as follows:

		\begin{equation}\label{e4}
		\begin{aligned}[l]
		&\alpha_n (k)''=\alpha_n (k)',\forall n\in \mathcal{L}, k=0,\ldots,\kappa, \\
		& p_{mm}(0)''=(p_{mm}(0)'x_{m,1}^\alpha(0)'+2e)/x_{m,1}^\alpha(0)',\\
		&p_{im}(0)''=(p_{im}(0)'x_{m,1}^\alpha(0)'-2e)/x_{m,1}^\alpha(0)',\\
		&p_{tm}(k)''= p_{nm}(k)', \forall t\in \{i,m\},k=1,\ldots, \kappa\,\\
		&p_{qm}(k)''= p_{qm}(k)' , \forall q\in \mathcal{L}\backslash\{i,m\}, k=0,\ldots,\kappa,\\
		&p_{np}(k)''=p_{np}(k)', \forall n\in \mathcal{L},\forall p \in \mathcal{L}\backslash\{m\},k=0,\ldots, \kappa.
		\end{aligned}
		\end{equation}
		
		Under the initial substate values in (\ref{sub}) and the weights \eqref{e4},
		it can be easily verified that the information set at $k=0$ equals to $\mathcal{I}_{\mathcal{A}} (0)$ and  at  $k=1$, $x_{q,1}^\alpha(1)''=x_{q,1}^\alpha(1)', \forall q\in \mathcal{L}_i\backslash\{i,m\}.$
		
		Also, we can obtain that
		\begin{equation*}
		\begin{aligned}[l]
		x_{i,1}^\alpha(1)''&=p_{im}(0)''x_{m,1}^\alpha(0)''\\&+\sum\limits_{j\in N^{in}_i\backslash\{m\}\cup\{i\}}p_{ij}(0)'' x^\alpha_{j,1}(0)''+ x^\beta_{i,1}(0)''\\
		&=(p_{im}(0)'x_{m,1}^\alpha(0)'-2e)x_{m,1}^\alpha(0)'/x_{m,1}^\alpha(0)'\\&+\sum\limits_{j\in N^{in}_i\backslash\{m\}\cup\{i\}}p_{ij}(0)' x^\alpha_{j,1}(0)'+x_{i,1}^\beta(0)'+2e\\
		&=p_{im}(0)'x_{m,1}^\alpha(0)'\\&+\sum\limits_{j\in N^{in}_i\backslash\{m\}\cup\{i\}}p_{ij}(0)' x^\alpha_{j,1}(0)'+x_{i,1}^\beta(0)'\\
		&= x_{i,1}^\alpha(1)',\\
		\end{aligned}
		\end{equation*}
		\begin{equation*}
		\begin{aligned}[l]
		x_{m,1}^\alpha(1)''&=p_{mm}(0)''x_{m,1}^\alpha(0)''\\&+\sum\limits_{j\in N^{in}_m}p_{mj}(0)'' x^\alpha_{j,1}(0)''+ x^\beta_{m,1}(0)''\\
		&=(p_{mm}(0)'x_{m,1}^\alpha(0)'+2e)x_{m,1}^\alpha(0)'/x_{m,1}^\alpha(0)'\\&+\sum\limits_{j\in N^{in}_m}p_{mj}(0)' x^\alpha_{j,1}(0)'+ x^\beta_{m,1}(0)'-2e\\
		&=p_{mm}(0)'x_{m,1}^\alpha(0)'\\&+\sum\limits_{j\in N^{in}_m}p_{mj}(0)' x^\alpha_{j,l}(0)'+ x^\beta_{m,1}(0)'\\
		&= x_{m,1}^\alpha(1)'.\\
		\end{aligned}
		\end{equation*}
		
		Thus, at $k=1$, we have $x_{n,1}^\alpha(1)''=x_{n,1}^\alpha(1)', \forall n\in \mathcal{L}.$ For $k \geq 1$,  it can be easily obtained that $ x_{n,1}^\alpha(k)''=x_{n,1}^\alpha(k)', \forall n\in \mathcal{L}$ under the weights \eqref{e4}. Then, it can be concluded that under the initial substate values in (\ref{sub}) and  weights \eqref{e4}, the information set sequence accessible to set $\mathcal{A}$ equals to $\mathcal{I}_{\mathcal{A} }(0:\kappa)$.

		Situation \uppercase\expandafter{\romannumeral 2}: If  $m\in N^{out}_i$, we consider  the following weights:
		
		\begin{equation}\label{e5}
		\begin{aligned}[l]
		&\alpha_n (k)''=\alpha_n (k)',\forall n\in \mathcal{L}, k=0,\ldots,\kappa,  \\
		&p_{ii}(0)''=(p_{ii}(0)'x_{i,1}^\alpha(0)'-2e)/x_{i,1}^\alpha(0)',\\
		& p_{mi}(0)''=(p_{mi}(0)'x_{i,1}^\alpha(0)'+2e)/x_{i,1}^\alpha(0)',\\
		& p_{ti}(k)''= p_{ti}(k)', \forall t\in \{j,m\},k=1,\ldots,\kappa,\\
		&  p_{qi}(k)''= p_{qi}(k)' , \forall q\in \mathcal{L}\backslash\{i,m\}, k=0,\ldots,\kappa,\\
		& p_{np}(k)''=p_{np}(k)', \forall n\in \mathcal{L},\forall p \in \mathcal{L}\backslash\{i\},k=0,\ldots,\kappa.
		\end{aligned}
		\end{equation}
		
		Under the initial substate values in (\ref{sub}) and the weights \eqref{e5},
		we can easily obtain  that the information set at $k=0$ equals to $\mathcal{I}_{\mathcal{A}} (0)$ and  at  $k=1$, $x_{q,1}^\alpha(1)''=x_{q,1}^\alpha(1)', \forall q\in \mathcal{L}\backslash\{i,m\}.$
		
		Moreover, we have that
		\begin{equation*}
		\begin{aligned}[l]x_{i,1}^\alpha(1)''&=p_{ii}(0)''x_{i,1}^\alpha(0)''\\&+\sum\limits_{j\in N^{in}_i}p_{ij}(0)'' x^\alpha_{j,1}(0)''+ x^\beta_{i,1}(0)''\\
		&=(p_{ii}(0)'x_{i,1}^\alpha(0)'-2e)x_{i,1}^\alpha(0)'/x_{i,1}^\alpha(0)'\\&+\sum\limits_{j\in N^{in}_i}p_{ij}(0)' x^\alpha_{j,1}(0)'+x_{i,1}^\beta(0)'+2e\\
		&=p_{ii}(0)'x_{i,1}^\alpha(0)'\\&+\sum\limits_{j\in N^{in}_i}p_{ij}(0)' x^\alpha_{j,1}(0)'+x_{i,1}^\beta(0)'\\
		&= x_{i,1}^\alpha(1)',\\
		\end{aligned}
		\end{equation*}
		\begin{equation*}
		\begin{aligned}[l]
		x_{m,1}^\alpha(1)''&=p_{mi}(0)''x_{i,1}^\alpha(0)''\\&+\sum\limits_{j\in N^{in}_m\backslash\{i\}\cup\{m\}}p_{mj}(0)'' x^\alpha_{j,1}(0)''+ x^\beta_{m,1}(0)''\\
		&=(p_{mi}(0)'x_{i,1}^\alpha(0)'+2e)x_{i,1}^\alpha(0)'/x_{i,1}^\alpha(0)'\\&+\sum\limits_{j\in N^{in}_m\backslash\{i\}\cup\{m\}}p_{mj}(0)' x^\alpha_{j,1}(0)'+ x^\beta_{m,1}(0)'-2e\\
		&=p_{mi}(0)'x_{i,1}^\alpha(0)'\\&+\sum\limits_{j\in N^{in}_m\backslash\{i\}\cup\{m\}}p_{mj}(0)' x^\alpha_{j,1}(0)'+ x^\beta_{m,1}(0)'\\
		&= x_{m,1}^\alpha(1)'.\\
		\end{aligned}
		\end{equation*}
		
		Hence, at $k=1$, we have $x_{n,1}^\alpha(1)''=x_{n,1}^\alpha(1)', \forall n\in \mathcal{L}.$ Similar to Situations \uppercase\expandafter{\romannumeral 1}, we can obtain  that under the initial substate values in (\ref{sub}) and  weights \eqref{e5}, the information set sequence accessible to set $\mathcal{A}$ equals to $\mathcal{I}_{\mathcal{A}}(0:\kappa)$.	
	
	{\color{black}Summarizing Situations \uppercase\expandafter{\romannumeral 1} and \uppercase\expandafter{\romannumeral 2}, we have that $ x_i(0)''=x_i(0)'+e\in \Delta(\mathcal{I}_\mathcal{A}(0:\kappa))$, then
		$$\text{Diam}(\mathcal{I}_\mathcal{A}(0:\kappa))\geq \sup\limits_{e\in\mathbb{R}} |x_i(0)'-(x_i(0)'+e)| = \sup\limits_{e\in\mathbb{R}} {|e|}=\infty.$$
		
		The above analysis holds for any $\kappa \in \mathbb{N}$ and any feasible $\mathcal{I}_\mathcal{A}(0:\kappa)$. }Therefore, by Definition 4, the privacy of node $i$ is preserved against a set of honest-but-curious nodes $\mathcal{A}$ if node $i$ has at least one legitimate neighbor node $m \notin \mathcal{A}$.
\end{proof}
{ 
	\begin{remark}
		The values of $x_{i,2}^\alpha(0)$ and $x_{i,2}^\beta(0)$ are known to set $\mathcal{A}$ since each node $q$ initializes $x_{q,2}^\alpha(0)=0$ and $x_{q,2}^\beta(0)=2$, $\forall q\in \mathcal{V}$. Nevertheless, the knowledge of $x_{i,2}^\alpha(0)$ and $x_{i,2}^\beta(0)$ does not help the adversary to infer $x_i(0)$ since the transmitted data $p_{ji}(0)x_{i,2}^\alpha(0), j\in N_i^{out}$ is equal to zero under any $p_{ji}(0)$.
\end{remark}}

Note that depending on the value of $e$, the weights $p_{mm}(0)'',  p_{jm}(0)'', p_{jj}(0)''$, $p_{mj}(0)''$ in (\ref{e4}) and  (\ref{e5}) could be outside the range $({\color{black}0},1)$. To ensure $ (x_i(0)'+e) \in \Delta(\mathcal{I}_\mathcal{A}(0:\kappa))$ for arbitrary real number $e$ , the weights at $k=0$ should be unrestricted, which is consistent with the weight selection in Step 1.2) of Algorithm \ref{alg2}.

\begin{theorem}
For a digraph satisfying \textit{Assumption \ref{asp1}}, the initial value of any node  $i$ can be uniquely inferred in an asymptotic sense by a set of honest-but-curious nodes $\mathcal{A}$ if $N_i^{out} \cup N_i^{in}\subseteq \mathcal{A}$.
\end{theorem}
\begin{proof}
	Since $N_i^{out} \cup N_i^{in}\in \mathcal{A}$, i.e., all the neighbors of node $i$ belong to set $\mathcal{A}$, the following equations can be obtained:
	{   \begin{equation} \label{e9}
		\begin{aligned}
		&z_l(k+1)-z_l(k)\\
		&=\sum\limits_{m\in N_i^{in}}p_{im}(k)x_{m,l}^\alpha(k)-\sum\limits_{n\in N_i^{out}}p_{ni}(k)x_{i,l}^\alpha(k),\\
		\end{aligned}		
		\end{equation}}where $z_l(k)=x_{i,l}^\alpha(k)+x_{i,l}^\beta(k)$, $l\in\{1,2\}$. With the information accessible to set $\mathcal{A}$ and  $z_2(0)=2$, $z_2(k)$ can be easily computed by: {    $$
		\begin{aligned}
		z_2(k)&=z_2(0)\\&+\sum\limits_{t=0}^{k-1}\bigg[\sum\limits_{m\in N_i^{in}}p_{im}(k)x_{m,2}^\alpha(k)-\sum\limits_{n\in N_i^{out}}p_{ni}(k)x_{i,2}^\alpha(k)\bigg].\\
		\end{aligned}
		$$}
	
	As $k$ goes infinity, average consensus will be achieved asymptotically, i.e.,$$\lim\limits_{k \rightarrow \infty } \frac{x^\alpha_{i,1}(k)} {x^\alpha_{i,2}(k)}=\lim\limits_{k \rightarrow \infty } \frac{x^\beta_{i,1}(k)} {x^\beta_{i,2}(k)}=\lim\limits_{k \rightarrow \infty } \frac{z_1(k)}{z_2(k)}, $$
	which leads to $\lim\limits_{k \rightarrow \infty } z_1(k)=\lim\limits_{k \rightarrow \infty } (z_2(k)x^\alpha_{i,1}(k))/(x^\alpha_{i,2}(k)).$
	
	Then, $z_1(0)$ can be obtained by
	{    $$
		\begin{aligned}
		z_1(0)=\lim\limits_{k \rightarrow \infty }\bigg( z_1(k)-\sum\limits_{t=0}^{k-1}\bigg[&\sum\limits_{m\in N_i^{in}}p_{im}(k)x_{m,1}^\alpha(k)\\&-\sum\limits_{n\in N_i^{out}}p_{ni}(k)x_{i,1}^\alpha(k)\bigg]
		\bigg).\\
		\end{aligned}
		$$}
	
	{ Therefore, $x_i(0)$ is uniquely estimated by set $\mathcal{A}$ in an asymptotic sense if $N_i^{out} \cup N_i^{in}\in \mathcal{A}$.}
\end{proof}
\begin{remark}
	{ Theorem 2 and Theorem 3 imply that the single neighbor configuration should be avoided in order to preserve privacy, which is also indicated in other privacy-preserving approaches such as \cite{gao2018privacy}, \cite{kempe2003gossip} and \cite{ruan2019secure}.}
\end{remark}
\subsection{Privacy-preserving Performance Analysis Against External Eavesdroppers}
{ In this subsection, we show that an eavesdropper defined in Definition \ref{df2} fails to estimate the initial value of each node when the system runs Algorithm \ref{alg2}. }
{ \begin{definition}
 The privacy of node $i$ is preserved against an eavesdropper defined in Definition 2 if there does not exist a natural number $k_1$ such that for all $k>k_1$, $|\hat x^0_i(k)-x_i(0)|<c$ for any $c>0$ with probability one.
\end{definition}

Intuitively, this means that there is no guarantee that the estimation error will be bounded by any finite $c$.} \begin{theorem} \label{th4}
	 For a digraph satisfying \textit{Assumption \ref{asp1}}, the privacy of node $i$ can be preserved by Algorithm \ref{alg2} {against an external eavesdropper defined in Definition \ref{df2}. }
\end{theorem}
\begin{proof}
The main idea of this proof is to show that {\color{black}there always exists a natural number $k_1$, such that  for all $k>k_1$, $|\hat x^0_i(k)-x_i(0)|>c$ for any $c>0$ with a nonzero probability.}
Without the knowledge of existing substates, an external eavesdropper assumes the exchange information  $p_{ij}(k)x_{j,l}^\alpha(k)$ as $p_{ij}(k)x_{j,l}^+(k)$. He then employs Algorithm \ref{alg3} to estimate the initial value of node $i$. The resulting estimation error is denoted as { $\left|\hat x^0_i(k)-x_i(0)\right|=\left|(a\beta_k+\gamma_k)/(2-\beta_k\right)|,$
where $\beta_k=\alpha_i(k-1) x_{i,2}^\alpha(k-1),\gamma_k=\bar x \alpha_i(k-1) x_{i,2}^\alpha(k-1)-\alpha_i(k-1) 
x_{i,1}^\alpha(k-1),$ $\bar x=\sum_{i=1}^N x_i(0)/N$ and $ a=x_i(0)-\bar x.$
Without loss of generality, we assume $a>0$. Otherwise one can replace $a$ and $\gamma_k$ by $-a$ and $-\gamma_k$, respectively, and follow the same analysis hereinafter. 

(1) From the convergence analysis in equation (\ref{cov1}), we have $$\lim\limits_{k\rightarrow \infty}\alpha_i(k-1) x_{i,1}^\alpha(k-1)= \bar x \alpha_i(k-1) x_{i,2}^\alpha(k-1),$$ i.e., $\lim\limits_{k\rightarrow \infty}\gamma_k=0$. {\color{black}Hence, there exists a natural number $k_1$ such that for all $k>k_1$, $| \gamma_k|<a.$  Then for any $k>k_1$ and $\beta_k \in (2-\delta_1, 2-\delta_2)$, where $0<\delta_2<\delta_1\leq \frac{a}{c+a},$ we have \begin{equation}\label{eq9}
|\hat x^0_i(k)-x_i(0)|>\left|(a\beta_k-a)/(2-\beta_k)\right|.
\end{equation}

Therefore, for $k>k_1$,  $$
\begin{aligned}
Pr&\Big(|\hat x^0_i(k)-x_i(0)|>c,\beta_k \in (2-\delta_1, 2-\delta_2)\Big)\\&
>Pr\bigg(\left|(a\beta_k-a)/(2-\beta_k) \right|>c ,\beta_k \in (2-\delta_1, 2-\delta_2)\Big).\\
\end{aligned}
$$ 

}
(2) If $\beta_k \in (2-\delta_1, 2-\delta_2)$, we can obtain $\left|(a\beta_k-a)/(2-\beta_k)\right|>c,$ i.e.,
{\color{black}
 $$Pr\Big(\left|(a\beta_k-a)/(2-\beta_k)\right|>c \Big| \beta_k \in (2-\delta_1, 2-\delta_2\Big)=1.$$
}
(3) Next we will prove 
\begin{equation}\label{eq11}
Pr\Big(x_{i,2}^\alpha(k-1)\geq 2-\delta_2\Big)> 0	
\end{equation}for any $k>1$ by induction. 

From the system dynamics in Algorithm \ref{alg2}, we have $x_{i,2}^\alpha(1)=2$ and $x_{i,2}^\alpha(2)=2\sum_{j\in N_i^{in}} p_{ij}(1)$. {\color{black}Since $p_{ij}(1)\in (0,1)$ for any $j\in N_i^{in}$,} one has $Pr\Big(\sum_{j\in N_i^{in}} p_{ij}(1)\geq (2-\delta_2)/2 \Big)> 0.$ Hence, (\ref{eq11}) holds when $k=2$ and $k=3$. 

Assuming that (\ref{eq11}) holds when $k=t$ and $k=t+1$, we have $x_{i,2}^\alpha(t+1)=\sum\limits_{j\in N^{in}_i\cup\{i\}}p_{ij}(t) x^\alpha_{j,2}(t)+\alpha_i(t-1)x^\alpha_{i,2}(t-1)\geq (2-\delta_2)(p_{ii}(t)+\alpha_i(t-1))$ with a nonzero probability. 

Since $p_{ii}(t)$ and $\alpha_i(t-1)$ are i.i.d. random variables over the range $(0, 1)$, one has $Pr((p_{ii}(t)+\alpha_i(t-1))\geq 1)> 0$. Hence, $Pr\Big(x_{i,2}^\alpha(k-1)\geq 2-\delta_2\Big)> 0$ holds for $k=t+2$. 

By induction, $Pr\Big(x_{i,1}^\alpha(k-1)\geq 2-\delta_2\Big)>0$ holds for any $k>1$.

(4) Since $\alpha_i(k-1)\in ({\color{black}0},1)$ and $Pr\Big(x_{i,2}^\alpha(k-1)\geq 2-\delta_2\Big)> 0$ for any $k>1$, {\color{black}we have $\forall k>1,$
\begin{equation}\label{eq12}
Pr\bigg(\alpha_i(k-1)\in\big(\frac{2-\delta_1}{x_{i,2}^\alpha(k-1)}, \frac{2-\delta_2}{x_{i,2}^\alpha(k-1)}\big)\bigg)>0,
\end{equation}
i.e., $Pr\big(\beta_{k}\in(2-\delta_1, 2-\delta_2)\big)> 0, \forall k>1$. }

Therefore, {\color{black}for all $k>k_1$,
$$	\begin{aligned}
		&Pr\Big(|\hat
		x^0_i(k)-x_i(0)|>c\Big)\\&\geq Pr\Big(|\hat x^0_i(k)-x_i(0)|>c,\beta_{k} \in (2-\delta_1, 2-\delta_2)\Big)\\&
		>Pr\bigg(\left|(a\beta_{k}-a)/(2-\beta_{k}) \right|>c ,\beta_{k} \in (2-\delta_1, 2-\delta_2)\Big)\\&=Pr\Big(\left|(a\beta_{k}-a)/(2-\beta_{k})\right|>c \Big| \beta_{k} \in (2-\delta_1, 2-\delta_2\Big)\\&\quad \cdot Pr\Big(\beta_{k} \in (2-\delta_1, 2-\delta_2)\Big)\\&=Pr\Big(\beta_{k} \in (2-\delta_1, 2-\delta_2)\Big)> 0,\\
	\end{aligned}$$} i.e., the privacy of node $i$ is preserved against an eavesdropper which runs Algorithm \ref{alg3}.}
\end{proof}


\section{SIMULATIONS}\label{sim}
In this section, we illustrate the effectiveness of our proposed algorithm and  show its advantages over existing privacy-preserving approaches on directed graphs. 

\begin{figure}[htp]
    \centering
    \includegraphics[width=0.2\textwidth]{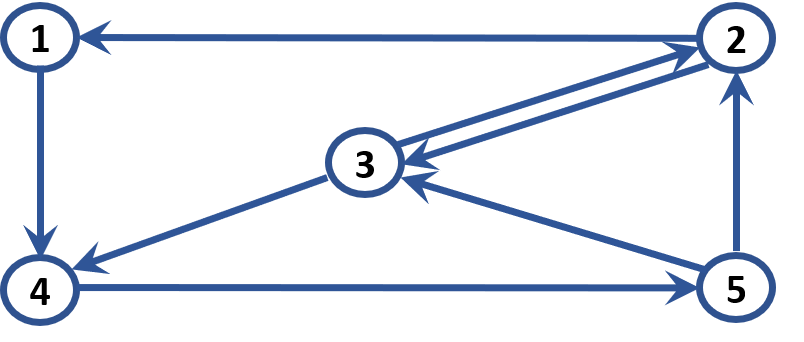}
    \caption{A strongly connected digraph of 5 nodes. }
    \label{digraph}

    \end{figure}

We first verify the convergence performance of our proposed algorithm communicating on a strongly connected digraph with $N=5$ nodes, shown in Fig.\ref{digraph}. The initial values for all nodes are chosen from $U(0, 50)$, $c$ is set to be $500$ and $M$ is set to be $100$. The evolution of the network and the convergence of $\hat x_i^{ave}$ are shown in Fig. \ref{decomp}. It is shown that the convergence is achieved, which is consistent with the theoretical results. {\color{black}Moreover, Fig. \ref{error} shows that the estimation error cannot be always bounded by $c$, i.e., the privacy of node $5$ is preserved against an external eavesdropper.}

\begin{figure}[htp]
    \centering
        \includegraphics[width=2.8in]{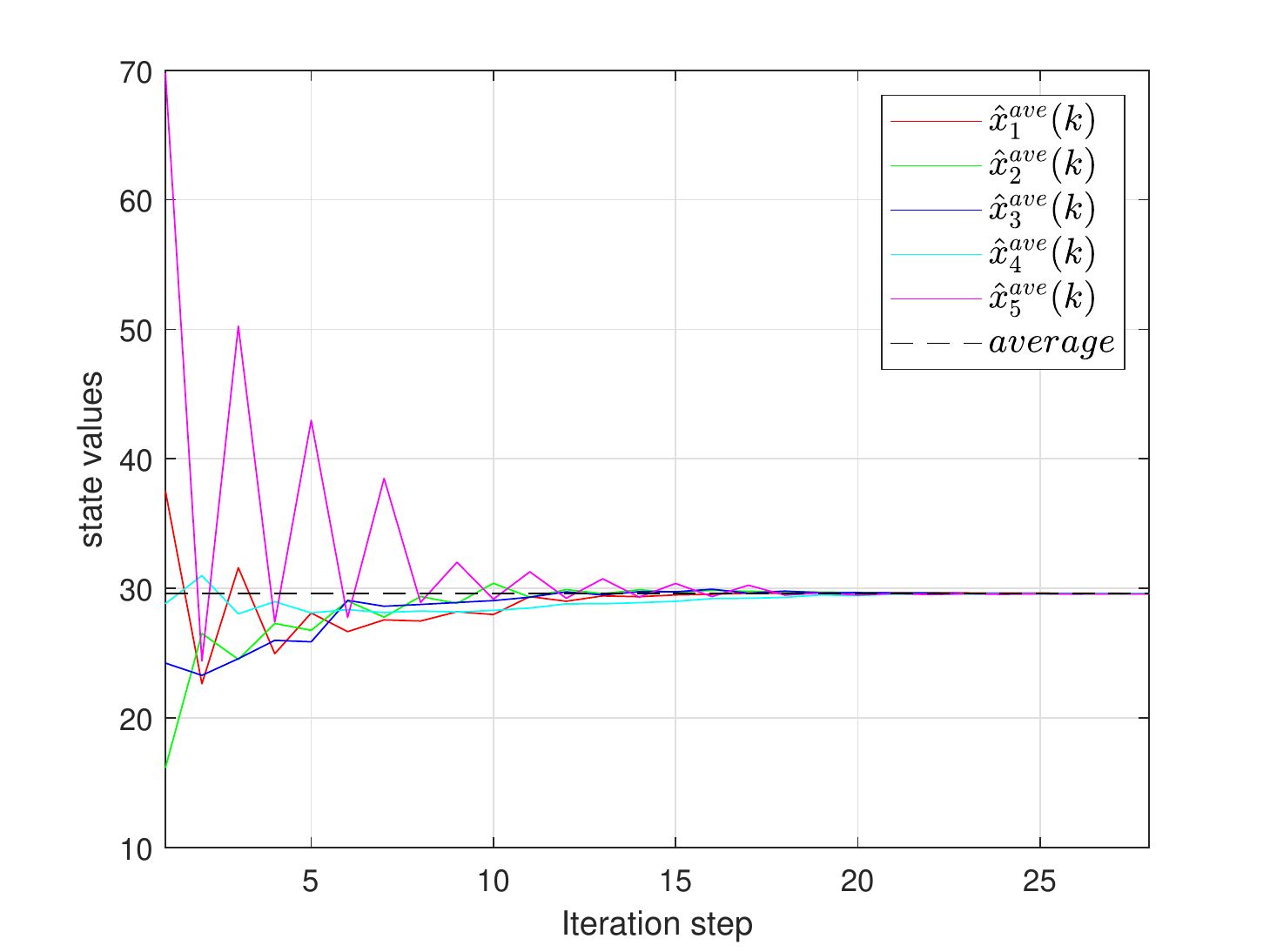}
          \vspace*{-2mm}
    \caption{{ The evolution of the estimated average obtained by the agents under Algorithm 3.}}
    \label{decomp}
    \end{figure}
     
\begin{figure}[htp]
    \centering
    \includegraphics[width=2.8in]{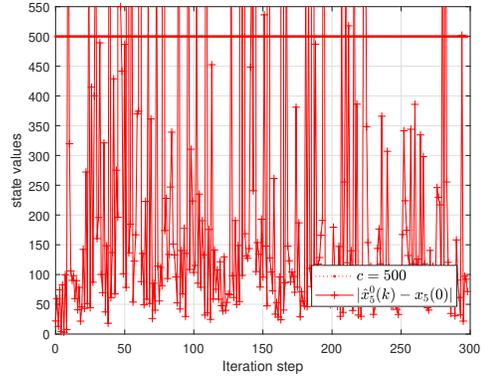}
    \caption{\color{black} The evolution of the absolute value of the error between the estimated $\hat x_5^0(k)$ and the true initial value $x_5(0)$  under Algorithm 2.}
    \label{error}
    \end{figure}

Unlike the privacy-preserving approach in \cite{charalambous2019privacy} and \cite{gao2018privacy}, where the convergence only happens after $k=L+1$ ($L$ is an randomly chosen integer), our proposed approach starts converging to the average at the iteration $k=1$. Fig. \ref{mse} shows the evolution of mean square error (MSE) under different privacy-preserving approaches, where $L$ is set to be $10$. According to Fig. \ref{mse}, our proposed approach converges faster than others and the MSE is much smaller before $k\leq L$. 

      \begin{figure}[htp]
    \centering
    \includegraphics[width=2.8in]{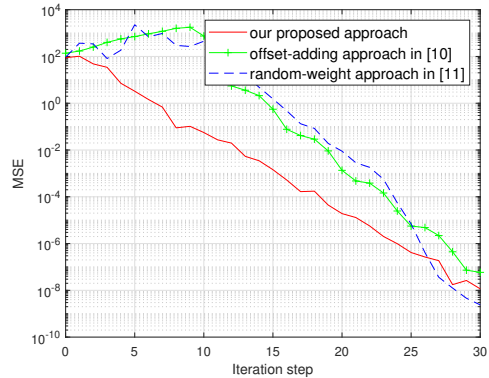}  
    \caption{ The evolution of MSE under different privacy-preserving approaches.}
    \label{mse}
    \end{figure}
    
Next, we investigate the privacy-preserving performance of our proposed approach. Without loss of generality, we assume that an eavesdropper is interested in the initial value of node 5 and applies Algorithm \ref{alg3} to estimate it. 
\begin{figure}[h]
    \centering
 \includegraphics[width=2.8in]{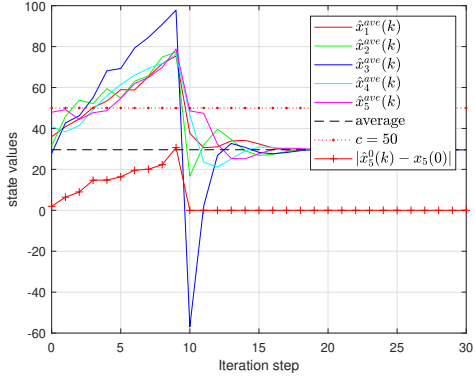}
    \caption{ Eavesdropper can estimate the initial state of node 5 under the privacy-preserving in \cite{charalambous2019privacy} by applying Algorithm \ref{alg3}. }
        \label{offset}
    \end{figure}

Fig. \ref{offset} shows that the evolution of estimated average and the eavesdropper's estimated value under the privacy-preserving approach in \cite{charalambous2019privacy}. It can be seen that as the network converges and the estimated average value reaches the true average value, the eavesdropper's estimated value also converges to the initial value of node 5. Similar results can be obtained using the random-weight approach in \cite{gao2018privacy}. { Therefore, the error between $\hat x_5(0)$ and $x_5(0)$ is bounded, i.e., the privacy of node $5$ cannot be preserved against an eavesdropper in \cite{charalambous2019privacy} and \cite{gao2018privacy}.}

\begin{table}[htp]
    \centering
\caption{Comparison among existing privacy-preserving approaches}
\label{table1}
\begin{tabular}{c|ccccc}
 & ours & \cite{charalambous2019privacy} &  \cite{gao2018privacy}&\cite{hadjicostis2020privacy}& \cite{wang2019privacy}\\ \hline
undirected graphs& \checkmark & \checkmark & \checkmark &\checkmark  &\checkmark  \\
digraphs& \checkmark & \checkmark &\checkmark  & \checkmark & $\times $\\
honest-but-curious nodes& \checkmark &\checkmark  & \checkmark &\checkmark &\checkmark \\
eavesdroppers&\checkmark & $\times $ & $\times $&\checkmark &\checkmark \\
complexity&low & low & low & high & low
\end{tabular}
\end{table}

It can be seen in Table \ref{table1} that, besides homomorphic encryption \cite{hadjicostis2020privacy}, our proposed algorithm can be effective in a number of situations. Furthermore, our proposed algorithm is built on the simple multiplications and addition steps (can be computed in $\mathcal{O}$(1) time). By comparison, the implementation of \cite{hadjicostis2020privacy} requires extra cryptosystems to perform complicated modular exponentiation steps (see \cite{goldreich2007foundations} for more details about its computation and time complexity). Hence, in contrast to privacy-preserving approach in \cite{hadjicostis2020privacy}, our proposed algorithm covers a wider range of applications.

\section {CONCLUSION AND FUTURE WORK}

In this paper, we proposed a privacy-preserving push-sum algorithm based on state decomposition for systems interacting on directed graphs. While protecting privacy from honest-but-curious nodes, our approach can guarantee the convergence to exact average. Moreover, in contrast to the offset-adding or random-weight approach, our proposed approach can prevent an external eavesdropper from estimating the initial value of each node. Furthermore, our proposed algorithm has lower computation and communication complexity which can be easily implemented in practice. Future work includes {\color{black}analyzing the convergence rate of the privacy-preserving push-sum algorithm}  and studying other types of consensus problems. 
\\


\end{document}